\makeatletter
\def\input@path{{styles/}}
\makeatother

\ifx\SOSA\undefined
    \def\UseBibLatex{1}
    \documentclass[11pt]{article}
    \newcommand{\notSOSA}[1]{#1}
    \newcommand{\SOSAVer}[1]{}
\else
    \documentclass[twoside,leqno]{article}
    \usepackage[letterpaper]{geometry}

    \usepackage{siamproceedings}
    \newcommand{\notSOSA}[1]{}
    \newcommand{\SOSAVer}[1]{#1}
    \newsiamremark{observation}{Observation}
    \newsiamthm{claim}{Claim}
    \newsiamremark{remark}{Rermark}

\fi

\providecommand{\BibLatexMode}[1]{}
\providecommand{\BibTexMode}[1]{}

\renewcommand{\BibLatexMode}[1]{#1}
\renewcommand{\BibTexMode}[1]{}

\ifx\UseBibLatex\undefined%
  \renewcommand{\BibLatexMode}[1]{}
  \renewcommand{\BibTexMode}[1]{#1}
\fi

\BibLatexMode{%
   \usepackage[bibencoding=utf8,style=alphabetic,backend=biber]{biblatex}%
   \usepackage{sariel_biblatex}%
}

\usepackage{amsmath}%
\usepackage{amssymb}%
\usepackage[table]{xcolor}%

\usepackage{mathcalb}%

\usepackage[T1]{fontenc}

\notSOSA{%
\usepackage[amsmath,thmmarks]{ntheorem}%

\usepackage{titlesec}%
}

\titleformat{\subsubsection}[runin]
  {\normalfont\normalsize\bfseries} %
  {\thesubsubsection.}               %
  {0.5em}                             %
  {}                                %
  [.]

\usepackage{xcolor}%
\usepackage{mleftright}%
\usepackage{xspace}%
\usepackage{float}
\usepackage{graphicx}
\usepackage{hyperref}%
\usepackage[inline]{enumitem}
\usepackage{hyperref}%
\usepackage[ocgcolorlinks]{ocgx2}
\usepackage{euscript}%
\usepackage{caption}

\usepackage{todonotes}

\usepackage[cm]{fullpage}%

\hypersetup{%
      unicode,
      breaklinks,%
      colorlinks=true,%
      urlcolor=[rgb]{0.25,0.0,0.0},%
      linkcolor=[rgb]{0.5,0.0,0.0},%
      citecolor=[rgb]{0,0.2,0.445},%
      filecolor=[rgb]{0,0,0.4},
      anchorcolor=[rgb]={0.0,0.1,0.2}%
}

\notSOSA{

   \titlelabel{\thetitle. }%

\theoremseparator{.}%

\theoremstyle{plain}%
\newtheorem{theorem}{Theorem}[section]

\newtheorem{lemma}[theorem]{Lemma}

{
\theoremseparator{}%

}

\theoremseparator{.}%

\newtheorem{claim}[theorem]{Claim}%

\newtheorem{observation}[theorem]{Observation}

\theoremstyle{plain}%
\theoremheaderfont{\sf} \theorembodyfont{\upshape}%
\newtheorem*{remark:unnumbered}[theorem]{Remark}%
\newtheorem{remark}[theorem]{Remark}%
\newtheorem{definition}[theorem]{Definition}

\theoremheaderfont{\em}%
\theorembodyfont{\upshape}%
\theoremstyle{nonumberplain}%
\theoremseparator{}%
\theoremsymbol{\myqedsymbol}%
\newtheorem{proof}{Proof:}%

}

\providecommand{\emphind}[1]{}%
\renewcommand{\emphind}[1]{\emph{#1}\index{#1}}

\definecolor{blue25emph}{rgb}{0, 0, 11}

\providecommand{\emphic}[2]{}
\renewcommand{\emphic}[2]{\textcolor{blue25emph}{%
      \textbf{\emph{#1}}}\index{#2}}

\providecommand{\emphi}[1]{}%
\renewcommand{\emphi}[1]{\emphic{#1}{#1}}

\definecolor{almostblack}{rgb}{0, 0, 0.3}

\providecommand{\emphw}[1]{}%
\renewcommand{\emphw}[1]{{\textcolor{almostblack}{\emph{#1}}}}%

\providecommand{\emphOnly}[1]{}%
\renewcommand{\emphOnly}[1]{\emph{\textcolor{blue25emph}{\textbf{#1}}}}

\newcommand{\myqedsymbol}{\rule{2mm}{2mm}} \newcommand{\EliotThanks}[1]{%
   \thanks{%
      School of Computing and Data Science; %
      University of Illinois; %
      201 N. Goodwin Avenue; %
      Urbana, IL, 61801, USA; %
      {\tt erobson2\atgen{}illinois.edu}; {\tt \url{https://eliotwrobson.github.io/}.} #1}} \newcommand{\SarielThanks}[1]{%
   \thanks{%
      School of Computing and Data Science; %
      University of Illinois; %
      201 N. Goodwin Avenue; %
      Urbana, IL, 61801, USA; %
      \href{mailto:spam@illinois.edu}{sariel@illinois.edu}; %
      \url{http://sarielhp.org/}.%
      #1%
   }%
} \newcommand{\BenThanks}[1]{%
   \thanks{Department of Computer Science; University of Texas at Dallas; Richardson, TX 75080, USA; {\tt benjamin.raichel\atgen{}utdallas.edu}; {\tt \url{http://utdallas.edu/\string~benjamin.raichel}.} #1}} \newcommand{\atgen}{\symbol{'100}}

\newcommand{\HLink}[2]{\hyperref[#2]{#1~\ref*{#2}}}
\newcommand{\HLinkSuffix}[3]{\hyperref[#2]{#1\ref*{#2}{#3}}}

\newcommand{\figlab}[1]{\label{fig:#1}}
\newcommand{\figref}[1]{\HLink{Figure}{fig:#1}}

\newcommand{\thmlab}[1]{{\label{theo:#1}}}
\newcommand{\thmref}[1]{\HLink{Theorem}{theo:#1}}

\newcommand{\obslab}[1]{\label{observation:#1}}
\newcommand{\obsref}[1]{\HLink{Observation}{observation:#1}}

\newcommand{\remlab}[1]{\label{rem:#1}}
\newcommand{\remref}[1]{\HLink{Remark}{rem:#1}}%

\newcommand{\lemlab}[1]{\label{lemma:#1}}
\newcommand{\lemref}[1]{\HLink{Lemma}{lemma:#1}}%

\providecommand{\deflab}[1]{\label{def:#1}}
\newcommand{\defref}[1]{\HLink{Definition}{def:#1}}

\newcommand{\defrefregY}[2]{\hyperref[def:#1]{{\textcolor{yellow}{#2}}}}

\definecolor{blackish}{rgb}{0.14, 0, 0.0}
\newcommand{\defrefY}[2]{%
   \textcolor{blackish}{%
      \renewcommand\color[2][]{}%
      \defrefregY{#1}{#2}%
   }%
}

\providecommand{\eqlab}[1]{}%
\renewcommand{\eqlab}[1]{\label{equation:#1}}
\newcommand{\Eqref}[1]{\HLinkSuffix{Eq.~(}{equation:#1}{)}}

\providecommand{\remove}[1]{}%
\newcommand{\SetY}[2]{\left\{ #1 \;\middle\vert\; #2 \right\}}

\newcommand{\pth}[1]{\mleft(#1\mright)}%

\newcommand{\Vertices}{\Mh{\mathsf{V}}}%
\newcommand{\VV}{\Vertices}%

\newcommand{\ceil}[1]{\mleft\lceil {#1} \mright\rceil}

\newcommand{\brc}[1]{\left\{ {#1} \right\}}

\newcommand{\cardin}[1]{\left\lvert {#1} \right\rvert}%

\renewcommand{\th}{th\xspace}
\newcommand{\ds}{\displaystyle}%

\renewcommand{\Re}{\mathbb{R}}%

\newlist{compactenumA}{enumerate}{5}%
\setlist[compactenumA]{topsep=0pt,itemsep=-1ex,partopsep=1ex,parsep=1ex,%
   label=(\Alph*)}%

\newlist{compactenuma}{enumerate}{5}%
\setlist[compactenuma]{topsep=0pt,itemsep=-1ex,partopsep=1ex,parsep=1ex,%
   label=(\alph*)}%

\newlist{compactenumI}{enumerate}{5}%
\setlist[compactenumI]{topsep=0pt,itemsep=-1ex,partopsep=1ex,parsep=1ex,%
   label=(\Roman*)}%

\newlist{compactenumi}{enumerate}{5}%
\setlist[compactenumi]{topsep=0pt,itemsep=-1ex,partopsep=1ex,parsep=1ex,%
   label=(\roman*)}%

\newlist{compactitem}{itemize}{5}%
\setlist[compactitem]{topsep=0pt,itemsep=-1ex,partopsep=1ex,parsep=1ex,%
   label=\ensuremath{\bullet}}%

\numberwithin{figure}{section}%
\numberwithin{table}{section}%
\numberwithin{equation}{section}%

\usepackage{stmaryrd}%
\providecommand{\IntRange}[1]{\mleft\llbracket #1 \mright\rrbracket}
\newcommand{\IRX}[1]{\IntRange{#1}}%

\newcommand{\Term}[1]{\textsf{#1}}%
\newcommand{\LSH}{\Term{LSH}\xspace}%
\newcommand{\NN}{\Term{NN}\xspace}%
\newcommand{\ANN}{\Term{ANN}\xspace}%
\newcommand{\DiskANN}{\Term{DiskANN}\xspace}%
\newcommand{\WSPD}{\Term{WSPD}\xspace}

\newcommand{\repX}[1]{\zeta^{}_{#1}}%

\DefineNamedColor{named}{AlgorithmColor}{cmyk}{0.07,0.90,0,0.34}

\newcommand{\SaveContent}[2]{%
   \expandafter\newcommand{#1}{#2}%
}

\providecommand{\Mh}[1]{{#1}}%

\newcommand{\dY}[2]{\left\| #1 - #2 \right\|}

\newcommand{\eps}{\varepsilon}

\newcommand{\FMS}{\EuScript{X}\index{metric space}}

\newcommand{\dC}{\mathcalb{d}}%
\newcommand{\DistChar}{\dC}%
\newcommand{\dmY}[2]{\DistChar\pth{#1,#2}}%
\newcommand{\ts}{\hspace{0.6pt}}

\newcommand{\Ground}{{U}}

\DeclareMathAlphabet{\mathpzc}{OT1}{pzc}{m}{it}

\newcommand{\WR}{\mathcal{W}}
\newcommand{\PB}{{Q}}%
\newcommand{\PC}{{R}}%

\newcommand{\diamC}{\nabla}%
\newcommand{\diamX}[1]{\diamC\pth{#1}}
\newcommand{\diamY}[2]{\diamC^{}_{\!#1}\pth{#2}}

\newcommand{\G}{\mathsf{G}}%
\renewcommand{\H}{\mathsf{H}}%

\providecommand{\Edges}{\Mh{\mathsf{E}}}%
\providecommand{\EE}{\Edges}%
\providecommand{\EdgesX}[1]{\Edges\pth{#1}}%
\providecommand{\EGX}[1]{\EdgesX{#1}}%

\newcommand{\dsY}[2]{\dmY{#1}{#2}}%

\newcommand{\seclab}[1]{\label{sec:#1}}
\newcommand{\secref}[1]{\HLink{Section}{sec:#1}}

\providecommand{\P}{\mathsf{P}}%
\renewcommand{\P}{\mathsf{P}}%

\providecommand{\Q}{\mathsf{P}}%
\renewcommand{\Q}{\mathsf{Q}}%

\newcommand{\nnY}[2]{\mathsf{nn}^{}_{#1}\pth{#2}}

\newcommand{\Packing}{\mathcal{N}}%

\newcommand{\clmlab}[1]{\label{claim:#1}}
\newcommand{\clmref}[1]{\HLink{Claim}{claim:#1}}
\providecommand{\etal}{et~al.\xspace}
\renewcommand{\etal}{et~al.\xspace}
\newcommand{\pair}{\mathcalb{p}}%

\newcommand{\spread}{\Psi}%
\newcommand{\spreadX}[1]{\spread\pth{#1}}%
\newcommand{\cpX}[1]{\mathrm{cp}\pth{#1}}

\newcommand{\dirEdgeY}[2]{#1 \rightarrow #2}%
\newcommand{\ballC}{\mathcalb{b}}%
\newcommand{\ballY}[2]{\ballC\pth{#1,#2}}%
\newcommand{\permut}[1]{\left\langle {#1} \right\rangle}
\newcommand{\tbllab}[1]{\label{table:#1}}
\newcommand{\tblref}[1]{\HLink{Table}{table:#1}}

\newcommand{\Eps}{\tfrac{1}{\eps}}%
\newcommand{\hide}[1]{}

\BibLatexMode{%
   \bibliography{ann_graph}%
}

\title{The Road to the Closest Point is Paved by Good Neighbors}

\author{%
   Sariel Har-Peled%
   \SarielThanks{%
   The work on this paper was partially supported by NSF AF award CCF-2317241.%
   }%
   \and%
   Benjamin Raichel%
   \BenThanks%
   {%
   The work on this paper was partially supported by NSF CCF award 2311179.}%
   \and%
   Eliot W. Robson%
   \EliotThanks{}%
}%

\date{\today}

\begin{document}

\maketitle

\SOSAVer{%
   \fancyfoot[R]{\scriptsize{Copyright \textcopyright\ 2026 by SIAM\\
         Unauthorized reproduction of this article is prohibited.}
         }
}%

\begin{abstract}
    Given a set $\P$ of $n$ points in $\Re^d$, and a parameter $\eps \in (0,1)$, we present a new construction of a directed graph $\G$, of size $O(n/\eps^d)$, such that $(1+\eps)$-\ANN queries can be answered by performing a greedy walk on $\G$, repeatedly moving to a neighbor that is (significantly) better than the current point. To the best of our knowledge, this is the first construction of a linear size with no dependency on the spread of the point set. The resulting query time, is $O( \eps^{-d} \log \spread)$, where $\spread$ is the spread of $\P$. The new construction is surprisingly simple and should be practical.
\end{abstract}

\section{Introduction}

A problem commonly encountered is \emph{nearest neighbor search} (aka \emph{proximity search}) -- given a finite set $\P$, endowed with a metric $\dC$, preprocess $\P$ such that given a query point $q$ one can quickly compute its nearest neighbor $\nnY{q}{\P} = \arg \min_{ p \in \P}\dmY{q}{p}$ in $\P$. This problem was studied extensively in the last 60 years. In high-dimensional Euclidean space, the \emph{exact} problem can not be solved faster than the time it takes to scan the input \cite{him-anntr-12}. Even in moderate dimensions (say four), exact data structures seem hopeless.
Thus, people turned to approximation, where one can return a point sufficiently close to the answer.

For low/moderate dimensions, data structures based on $kd$-trees perform well for the \ANN (i.e., \emph{approximate nearest-neighbor}) problem, both in theory and practice \cite{amnsw-oaann-98}. The problem is significantly more challenging in higher dimensions, and even getting a data structure with sublinear query time is not easy.
Locality-sensitive hashing (\LSH) introduced by Indyk and Motwani \cite{im-anntr-98, him-anntr-12} offered a data structure that performs well in theory and practice.

\paragraph{\NN graph.}
Another natural approach is constructing a graph on the points of $\P$. Then perform an $A^*$-type search for the nearest neighbor, walking on the graph towards the closest point to the query.
Arya and Mount \cite{am-annqf-93} and Clarkson \cite{c-aacpq-94} both offered results along this direction, see \tblref{results} for details. This research direction was abandoned in theory because of better theoretical results \cite{amnsw-oaann-98}, but empirical work using this technique continued.

A desired property of these graphs is that \emphw{greedy routing} suffices -- that is, one starts with an arbitrary vertex, and performs a walk always moving to the neighbor of the current vertex closest to the query point (a discrete analogue of gradient descent), till convergence, and this yields the desired \ANN.

\paragraph*{Navigable graphs.}
For $\alpha >1$, a graph is \emphw{$\alpha$-navigable} if for any pair $s,t$, either $\dirEdgeY{s}{t} \in \EGX{\G}$, or there exist $\dirEdgeY{s}{y} \in \EGX{\G}$, such that $\dmY{y}{t} < \tfrac{1}{\alpha} \dY{s}{t}$. Namely, a neighbor of $s$ is ``significantly'' closer to the destination $t$.
Indyk and Xu \cite{ix-wcppa-23} showed that greedy routing on $\alpha$-navigable graph answers $\gamma$-\ANN queries, where $\gamma \approx \tfrac{\alpha+1}{\alpha-1}$. This ratio was improved to $1 + \tfrac{1}{\alpha-1}$ by Gollapudi \etal \cite{gksw-sbypi-25}.

\paragraph*{\DiskANN.}
Recently, Subramanya \etal~\cite{sdhkk-dfabp-19} (\DiskANN) presented results that seem to outperform existing techniques in practice. \DiskANN presents a rather interesting, but challenging to analyze, construction of the \NN-graph.
It starts with a random graph over the points.
In the cleaning stage, the algorithm randomly permutes the graph's vertices.
For each vertex in the permutation, it computes its $k$ nearest-neighbors, performing an $A^*$-type search in the existing graph. The algorithm then adds edges from these $k$-nearest-neighbors to the query vertex, repeating this for all vertices in the permutation. The algorithm repeats this cleaning process twice.

During this process, the algorithm prunes edges whenever a vertex $v$ outdegree exceeds a certain threshold $R$. The pruning for a vertex $v$, and its outgoing neighbors $\Gamma \subseteq \P$ in the graph, is done as follows. The algorithm repeatedly marks the closest point $u$ in $\Gamma$ to $v$. It then throws away all the points in $N$ that are ``sufficiently'' close to $u$ (including $u$ itself).
It repeats this process till $\Gamma$ is exhausted, or $R$ points are marked.
The algorithm then deletes all the outgoing unmarked edges from $v$ (i.e., only the marked neighbors survive).

Somewhat related approaches used in practice include \Term{H{NS}W} \cite{my-erann-20}, and \Term{N{S}G} \cite{fxwc-fanns-19}.  Indyk and Xu \cite{ix-wcppa-23} studied a variant of \DiskANN and provided theoretical analysis for its performance -- showing that it works well if the data is low-dimensional and of bounded spread. This slow preprocessing version performs a cleanup for each vertex in the graph separately, starting with the whole point set. They also showed a matching lower-bound showing that in the worst case, \DiskANN (with ``fast preprocessing'') needs linear query time.

\paragraph*{To spread or not to spread?}

Traditionally, there is a dislike for theoretical results that depend on the spread of the input. As a reminder, for a set $\P$ in a metric space, its spread $\spread = \diamX{\P} / \cpX{\P}$ -- that is, the ratio between the longest distance and the smallest distance between any two points of $\P$. Generally speaking, any dependency of the form $\log \spread$ in results can be replaced, usually after tedious and involved work, with $\log n$, where $n = \cardin{\P}$ \cite{hpm-fcnldm-05}. In practice, even in moderate dimensions, frequently the spread $\spread$ is small. Thus, logarithmic dependency on the spread is quite acceptable, and in some cases, even preferable \cite{ix-wcppa-23} to logarithmic dependency on $n$.

\paragraph*{\WSPD.}
In 1995, Callahan and Kosaraju \cite{ck-dmpsa-95} show that the Euclidean metric, for a set $\P$ of $n$ points in $\Re^d$, can be compactly described as the union of $O(n/\eps^d)$ bicliques, where all the distances of edges in a single biclique are the same up to a factor of $1\pm \eps$. Elegantly, the biclique cover is computed in $O(n \log n + n/\eps^d)$ time, with each biclique being represented as a pair of nodes in a constructed tree over the point set, see \secref{wspd_def} for details.

\paragraph*{Greedy permutation.}

Given a set $\P$ of $n$ points in a metric space, a natural way of ordering the points is provided by starting with an arbitrary point of $\P$, and then repeatedly picking the furthest point in $\P$ from the set of points picked so far. The resulting ordering of $\P = \permut{p_1, \ldots, p_n}$ is known as the \emphw{greedy permutation} \cite{h-gaa-11}, see \secref{greedy_permutation} for details. The greedy permutation can be approximated in near-linear time if the dimension of the metric space is low. It has the desired property that for any $k \in \{1,\ldots, n\}$, the prefix $\P_k = \{ p_1, \ldots, p_k\}$ is a $2$-approximation to the optimal $k$-center clustering of $\P$ \cite{g-cmmid-85}.

\begin{table}[t]
    \centering%
    \begin{tabular}{|c|c|l|l|}
      \hline
      Space
      &
        Query time
      &
        Ref
      &
        Remark
      \\
      \hline
      \hline
      $O\pth{ \frac{n}{\eps^{d-1}} \log n}$
      &
        $O( \frac{1}{\eps^{d-1}}\log^3 n )\Bigr.$
      &
        \cite{am-annqf-93}
      &
        Yao graph + skip-list.
      \\
      \hline
      $O\pth{ \frac{n}{\eps^{(d-1)/2}} \log \spread}\Bigr.$
      &
        $O( \frac{1}{\eps^{(d-1)/2}} \log \spread \cdot \log n)$
      &
        \cite{c-aacpq-94}
      &
        Opt approx Voronoi cells + skip-list
      \\
      \hline
      \hline
      $O \pth{ \frac{n}{\eps^d}  \log \spread }$
      &
        $O \bigl({ \frac{1}{\eps^d \log (1/\eps)}  \log^2 \spread }\bigr)\Bigr.$
      &
        \cite{ix-wcppa-23}
      &
        Analyzing \DiskANN \cite{sdhkk-dfabp-19}.
      \\
      \hline
      \hline
      $O\pth{\frac{n}{\eps^d} \log \spread }\Bigr.$
      &
        $O\bigl({\frac{1}{\eps^d  \log (1/\eps)} \log^2 \spread }\bigr)\Bigr.$
      &
        \lemref{basic_wspd_search}%
      & \WSPD based.
      \\
      \hline
      $O\pth{\frac{n}{\eps^d} \log \spread }\Bigr.$
      &
        $O\bigl({\frac{1}{\eps^d  } \log \spread + \log^2 \spread  }\bigr)\Bigr.$
      &
        \thmref{main_wspd}%
      & Uses two graphs.
      \\
      \hline
      $O\pth{\frac{n}{\eps^d} \log \spread }\Bigr.$
      &
        $O\bigl({\frac{1}{\eps^d  } \log \frac{1}{\eps} + \log \spread  }\bigr)\Bigr.$
      &
        \lemref{multi:graphs}%
      & Uses multiple graphs.
      \\
      \hline
      \hline
      $O\pth{\frac{n}{\eps^d}}\Bigr.$
      &
        $O\pth{\frac{1}{\eps^d}\log \spread }$
        &
          \thmref{a_n_n_main}%
      &
        Greedy permutation.
      \\
      \hline
    \end{tabular}
    \caption{Known results on \ANN via walks in a graph. The input is a set of $n$ points in $\Re^d$, and $\eps \in (0,1)$ is a parameter. The result returned in $(1+\eps)$-\ANN. All new results also hold for spaces with bounded doubling dimension.}
    \tbllab{results}
\end{table}

\subsection*{Our results}

We develop better guaranteed constructions of navigable graphs in low dimensions. Specifically, the input is a set $\P$ of $n$ points in $\Re^d$, and a parameter $\eps \in (0,1)$. (Our results also hold verbatim when $O(d)$ is the doubling dimension of the metric space hosting $\P$.) We show the following:

\begin{compactenumI}
    \smallskip%
    \item \textsf{\NN graph using \WSPD.}
    Inspired by the analysis of Indyk and Xu \cite{ix-wcppa-23}, we show how to construct a graph that can be used to answer $(1+\eps)$-\ANN by performing a greedy walk. Our construction uses \WSPD, and it intuitively provides a direct construction of a graph similar to the one built by \DiskANN for the settings analyzed by Indyk and Xu. The construction can be interpreted as providing an alternative explanation for the graph constructed by \DiskANN (when using ``slow-preprocessing''). The resulting graph size depends logarithmically on the spread of the input, see \tblref{results} for details.

    \smallskip%
    \item \textsf{\NN graph using greedy permutation.} %
    We provide a new construction for navigable graphs that uses the greedy permutation -- it connects $O(1/\eps^d)$ edges into a point in the permutation from previous points. Thus, the resulting graph has a size that is linear and independent of the spread of the point set. To our knowledge, this is the first construction to have this property. In addition, it does not use any Euclidean space properties, and applies to doubling spaces, unlike the constructions of Arya and Mount and the one by Clarkson. The query time is $O( \tfrac{1}{\eps^d} \log \spread)$, see
    \thmref{a_n_n_main} for details.
\end{compactenumI}

\paragraph{Paper organization.}
We provide some necessary background in \secref{prelim}.  \secref{n_n_graphs} describes some key components of \DiskANN.  \secref{wspd_def} describes \WSPD in detail.  \secref{n_n_wspd} describes the construction of navigable graph using \WSPD. \secref{n_n_greedy} describes the new construction using greedy permutation.

\section{Preliminaries}
\seclab{prelim}

\subsection{Metric spaces}
\seclab{metric}

\begin{definition}
    \deflab{metric_space_def}%
    A \emphi{metric space} $\FMS$ is a pair $\FMS = (\Ground, \DistChar )$, where $\Ground$ is the ground set, and $\DistChar: \Ground \times \Ground \rightarrow [0, \infty)$ is a \emphi{metric} satisfying the conditions: (i) $\dmY{x}{y} = 0$ if and only if $x =y$, (ii) $\dmY{x}{y} = \dmY{y}{x}$, and (iii) $\dmY{x}{y} + \dmY{y}{z} \geq \dmY{x}{z}$ (triangle inequality).
\end{definition}

\begin{definition}
    \deflab{spread}%
    For a set $\P \subseteq \Ground$, its \emphi{diameter} is $\diamY{\dC}{\P} = \max_{x,y \in \P} \dmY{x}{y}$. Its \emphi{closest pair distance} is $\cpX{\P} = \min_{x,y \in \P: x \neq y} \dmY{x}{y}$. The ratio between these two quantities is the \emphi{spread}: $\spreadX{\P} = \diamX{\P} / \cpX{\P}$.
\end{definition}

\begin{definition}
    For a point $q \in \Ground$, and a set $\P \subseteq \Ground$, the \emphi{nearest-neighbor} of $q$ in $\P$, is the point
    \begin{math}
        \nnY{\P}{q} = \arg \min_{p \in \P} \dmY{q}{p}.
    \end{math}
    The distance between $q$ and $\nnY{\P}{q}$ is denoted by $\dmY{q}{\P} = \min_{p \in \P} \dmY{q}{p}$.
\end{definition}

\begin{definition}
    \deflab{packing}%
    Consider a metric space $(\Ground, \DistChar)$, and a set $\P \subseteq \Ground$.  A set $\Packing \subseteq \P$ is an \emphi{$r$-packing} for $\P$ if the following hold:
    \begin{compactenumi}
        \smallskip%
        \item \emphw{Covering property}: All the points of $\P$ are within a distance $< r$ from the points of $\Packing$. Formally, for all $p \in \P$, $\dmY{p}{\Packing} <r$.

        \smallskip%
        \item \emphw{Separation property}: For any pair of points $x, y \in \Packing$, we have that $\dmY{x}{y} \geq r$.
    \end{compactenumi}
\end{definition}
The naive algorithm for computing a packing repeatedly marks any point in $\P$ at a distance $\geq r$ from the current marked points till no such point exists. The set of marked points forms an $r$-packing. Faster algorithms are known in some cases \cite{hr-nplta-15,ehs-agcds-20}.

\begin{definition}
    For a point $x \in \Ground$, and a radius $r \geq 0$, the \emphi{ball} of radius $r$ centered at $x$ is the set
    \begin{math}
        \ballY{x}{r} = \SetY{z \in \Ground}{\dmY{x}{z} \leq r}.
    \end{math}
\end{definition}

\begin{definition}
    For $\eps \in (0,1)$, and a query point $q \in \Ground$, a point $p$ is \emphi{$(1+\eps)$-\ANN} (approximate nearest-neighbor) for $q$ if $\dmY{q}{p} \leq (1+\eps) \dmY{q}{\P}$.
\end{definition}

\newcommand{\dblConst}{\lambda}
\newcommand{\dblDim}{\delta}

\paragraph*{Doubling metrics.}

Consider a finite metric space $\FMS = (\Ground, \DistChar )$, The \emphi{doubling constant} $\dblConst$ of a set $\Ground$, is the minimum integer $\dblConst > 0$, such that for every ball $\ballC$ of $\FMS$, can be covered by at most $\dblConst$ balls of at most half the radius. The \emphi{doubling dimension} of the metric space, denoted by $\dblDim$, is $\ceil{ \log_2 \dblConst }$. it is not hard to verify that $\Re^d$ has doubling constant $2^{O(d)}$, and thus doubling dimension $O(d)$. Doubling dimension is an abstraction of the standard Euclidean dimension. In many cases, real data has a much lower doubling dimension than the high-dimensional ambient space it lies in. Many algorithms, for low-dimensional Euclidean input, extend to spaces with low doubling dimension \cite{hm-fcnld-06}.

\subsection{Background on graph-based search for \ANN}
\seclab{n_n_graphs}

\subsubsection{Search procedure}
Consider a directed graph $\G = (\P, \EE)$ built over a set $\P$ of $n$ points in some metric space. The task is to compute the \ANN (or $k$ closest such points) for a given query point $q$. The algorithm performs a Dijkstra-like exploration of the graph -- specifically, one initializes the queue to contain some arbitrary start vertex $s$. Now, in each iteration, one extracts the minimum distance point in the queue from $q$, and adds all its outgoing neighbors, seen for the first time, to the queue. If the queue size exceeds a certain threshold $L$, one removes all the points from the queue except the $L$ closest to $q$, where $L$ is some prespecified parameter. As in Dijkstra, the algorithm avoids visiting the same node more than once. Once the queue is empty, the search is completed.  The procedure returns the $k$ closest vertices visited by the algorithm during this process, sorted by their distance from $q$.

\subsubsection{Greedy routing}
\seclab{greedy_routing}

A more straightforward search procedure performs a walk in the graph starting from a vertex. It repeatedly moves to a neighbor closer to the query point, till reaching a (usually approximate) local minimum. There are two natural variants:
\begin{compactenumA}
    \smallskip%
    \item The ``impulsive'' version moves as soon as a neighbor, significantly closer to the query, is encountered.

    \smallskip%
    \item The ``mature'' alternative moves to the best neighbor attached to the current point. The above search algorithm achieves this behavior if one sets $L=1$.
\end{compactenumA}

\subsubsection{Robust prune (\DiskANN)}
A key component is pruning the outgoing edges from vertices with high outdegree. So consider a vertex $v$ and its list of outgoing neighbors $N_v$. The idea is to prune away neighbors that are too close together. To this end, one sorts the points of $N_v$ by increasing distance from $v$, and let $N = \permut{p_1, \ldots, p_m}$ be the resulting ordered list. The algorithm repeatedly takes the first point $p$ from $N$, adds it to the output list $O_v$ (initially empty), and removes all the points of $N$ that are inside the ball
\begin{equation*}
    B_{\dirEdgeY{v}{p}} = \SetY{ f \in N}{ \alpha \dmY{p}{f} <  \dmY{v}{f}},
\end{equation*}
where $\alpha >1$ is some parameter (e.g., $\alpha=2$). In words, the set $B_{\dirEdgeY{v}{p}}$ contains all the points of $N$ that are $\alpha$-times closer to $p$ than to $v$. Intuitively, $p$ serves as a local distribution \emphw{center} for $v$ for all the points in $B_{\dirEdgeY{v}{p}}$.  In Euclidean space the loci of all points that are $\alpha$-times closer to $p$ than $v$ is an \emphw{Apollonius ball}, with center at
\begin{equation*}
    p +
    \frac{1}{\alpha^2 - 1}  (p-v),
    \qquad\text{and of radius}\qquad
    r
    =
    \frac{\alpha}{\alpha^2 - 1} \dY{v}{p},
\end{equation*}
see \lemref{Apollonius}. The algorithm removes $B_{\dirEdgeY{v}{p}}$ from $N$, and repeats the process till $N$ is exhausted. One then sets the outgoing edges from $v$ to the (hopefully) reduced list of centers selected -- that is, the edges added to $v$ are $\SetY{\dirEdgeY{v}{u}}{u \in O_v}$.  \figref{prune} shows the example of the output of this process.

\begin{figure}
    \includegraphics[page=2,width=0.47\linewidth]{figs/prune}
    \hfill%
    \includegraphics[page=1,width=0.47\linewidth]{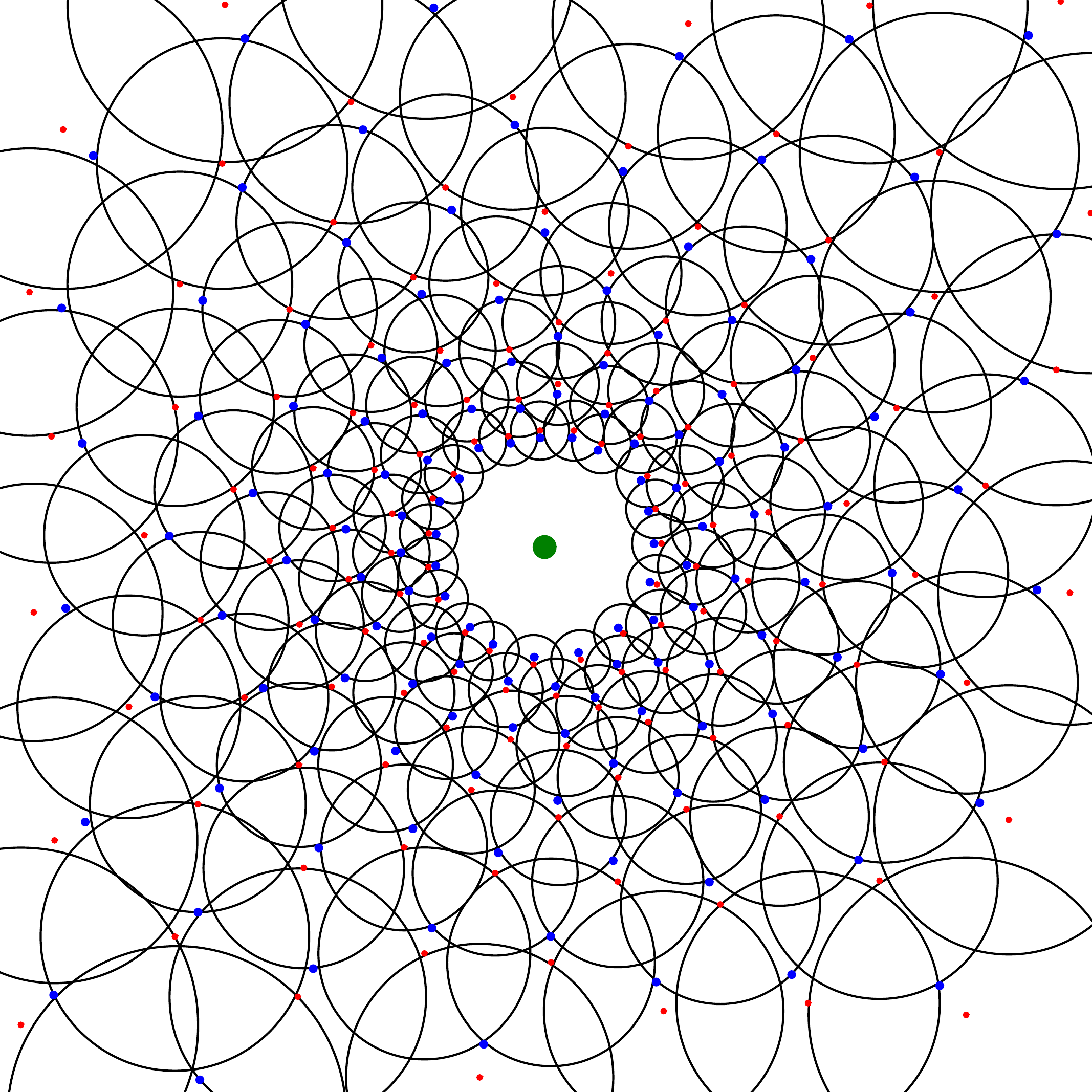}
    \caption{Left: The points selected by robust prune, with $\alpha=4$, where the original set of $\approx 200,000$ points is uniformly distributed in the square, except for a disallowed ``island'' in the middle. Right: The Apollonius disks that were used during this process. (We have not shown the original point set, as it simply forms a solid blob, and that seemed pointless [or is it point{}full?].)}
    \figlab{prune}
\end{figure}

Indyk and Xu \cite{ix-wcppa-23} showed that starting with $N_v = \P$, and performing this pruning for all the vertices of $\G$, the resulting graph answers $\gamma$-\ANN queries using greedy routing, where $\gamma \approx \frac{\alpha+1}{\alpha-1}$. (The version using $N_v = \P$ is the ``slow-preprocessing'' variant of \DiskANN.)
\begin{observation}
    \obslab{separated}%
    We are interested in the $(1+\eps)$-\ANN regime. That is $\gamma = 1+\eps$, for some $\eps \in (0,1)$. We thus have that $\alpha = 2/\eps +1$ in this case. The algorithm does pruning for the edge $\dirEdgeY{v}{p}$, and we get that the Apollonius ball in this case has its center close to $p$, and it has radius $\approx (\eps/2) \dY{v}{p}$. More precisely, the center is at
    \begin{math}
        p + \frac{\eps^2}{4(1+\eps)} (p-v)
    \end{math}
    and the radius is
    \begin{math}
        r
        =
        \pth{1+\frac{\eps}{2}}\frac{\eps}{2(1+\eps)} \dY{v}{p}.
    \end{math}
\end{observation}

\subsection{Background on \WSPD}
\seclab{wspd_def}

For a graph $\G= (\VV, \EE)$, and a set $Y \subseteq \VV$, the \emphi{induced subgraph} of $\G$ by $Y$ is
\begin{equation*}
    \G_Y = \bigl(Y, \SetY{ uv \in \EE}{ u,v \in Y} \bigr).
\end{equation*}
In the following, assume we are given a metric space $(\Ground,\DistChar)$.
For a detailed description of \WSPD and their construction algorithm, see  Har-Peled \cite{h-gaa-11}.

\begin{definition}
    For two sets $B,C \subseteq \Ground$, let
    \begin{math}
        B \otimes C =%
        \SetY{ bc}{b \in B, c \in C, b \neq c}.
    \end{math}

\end{definition}

\begin{definition}
    \deflab{pair_decomposition}%
    For a point set $\P \subseteq \Ground$, a \emphi{pair decomposition} of $\P$ is a set of pairs
    \[
        \WR = \brc{\bigl. \brc{A_1,B_1},\ldots,\brc{A_s,B_s}},
    \]
    such that
    \begin{enumerate*}[label=(\Roman*)]
        \item $A_i,B_i\subset \P$ for every $i$,
        \item $A_i \cap B_i = \emptyset$ for every $i$, and
        \item $\bigcup_{i=1}^s A_i \otimes B_i = \binom{\P}{2} = \P \otimes \P$.
    \end{enumerate*}
\end{definition}

\begin{definition}
    \deflab{well_separated}%
    The pair $\{\PB, \PC\}$ is \emphi{$\tfrac{1}{\eps}$-separated} by $\DistChar$ if
    \begin{equation*}
        \max \pth{\bigl. \diamY{\DistChar}{\PB},
           \diamY{\dC}{\PC} } \leq \eps \ts \dmY{\PB}{\PC},
        \qquad\text{where}\qquad%
        \dmY{\PB}{\PC} = \ds \min_{x \in \PB, y \in \PC} \dmY{x}{y}.
    \end{equation*}
\end{definition}

\begin{definition}
    \deflab{WSPD}%
    For a point set $\P$, a \emphOnly{well-separated pair decomposition} of $\P$ with parameter $1/\eps$, denoted by \emphw{$\tfrac{1}{\eps}$-\WSPD{}}, is a pair decomposition
    \begin{math}
        \WR = \brc{\bigl.  \brc{A_1,B_1},\ldots,\brc{A_s,B_s}}
    \end{math}
    of $\P$, such that, for all $i$, the sets $A_i$ and $B_i$ are $\tfrac{1}{\eps}$-separated.
\end{definition}

\begin{theorem}[\cite{ck-dmpsa-95}]
    \thmlab{WSPD}%
    For $\eps \in (0,1)$, and a set $\P$ of $n$ points in $\Re^d$, one can construct, in $O \bigl( n \log n + {n}/{ \eps^{d}} \bigr)$ time, an $\tfrac{1}{\eps}$-\WSPD of $\P$ of size $O(n/{ \eps^{d}})$.
\end{theorem}

\begin{remark}
    A similar result to \thmref{WSPD} is known for doubling metrics \cite{hm-fcnld-06}. Formally, for a point set $\P$ in a metric with doubling dimension $d$, one can compute a \WSPD of $\P$ of size $n/\eps^{O(d)}$ in $O \pth{ n \log n + {n}/{ \eps^{O(d)}}}$ time.
\end{remark}

For a pair $\pair = \{B,C \} \in \WR$, its \emphi{diameter} is $\diamX{\pair} = \diamX{ B \cup C}$.

\section{Nearest-neighbor graph via \WSPD}
\seclab{n_n_wspd}

\obsref{separated} points out that the Apollonius ball, constructed for the edge $\dirEdgeY{v}{p}$, used to prune away ``useless'' neighbors of $v$ near $p$, is $\tfrac{1}{\eps}$-well-separated from $v$. Namely, $v$ should have an outgoing edge for each \WSPD pair, say $\{B,C\}$, that contains it (say $v \in B$), to some representative $\repX{C} \in C$ (i.e., the edge is $\dirEdgeY{v}{\repX{C}}$).  This idea gives rise to a direct construction of a navigable graph.

\begin{remark}
    \remlab{search}%
    In the algorithm described next, the greedy routing always picks the minimum outgoing neighbor as the next vertex to use in the search. In addition, the search procedure stops as soon as the improvement in the distance to the query is insignificant in a round. Formally, if $\ell_i$ and $\ell_{i+1}$ are the distances from the query point to two consecutive vertices in the greedy routing, the algorithm stops if $\ell_{i+1} \geq (1-\eps/4)\ell_i$.
\end{remark}

\begin{lemma}
    \lemlab{basic_wspd_search}%
    Let $\P$ be a set of $n$ points in $\Re^d$, and assume $\P$ has \defrefY{spread}{spread} $\spread = \spreadX{\P}$. Then, for a prespecified parameter $\eps \in (0,1)$, one can construct a directed graph $\G$ over $\P$, that has $O( \eps^{-d} n \log \spread )$ edges, and the greedy routing on $\G$, answers $(1+\eps)$-\ANN queries, in $O( \log \spread )$ steps, and $O( \frac{1}{\eps^d \log (1/\eps)} \log^2 \spread)$ time.
\end{lemma}
\begin{proof}
    Let $\WR$ be $\tfrac{8}{\eps}$-\WSPD of $\P$ computed in $O( n \log n + n/\eps^d)$ time. Every point of $\P$ participates in at most $O(\tfrac{1}{\eps^d} \log \spread)$ pairs in the \WSPD \cite{h-gaa-11}.  In addition, for any set $X \in \{A, B\} \in \WR$, there is a precomputed representative $\repX{X} \in X$. In particular, for all pairs $\{B,C\} \in \WR$, consider the set of edges
    \begin{equation*}
        \EE(B,C)
        =
        \SetY{ \dirEdgeY{c}{\repX{B}} }{ c \in C} \cup
        \SetY{ \dirEdgeY{b}{\repX{C}} }{ b \in B}.
    \end{equation*}
    Let $\EE$ be the union of all such sets. Clearly,
    \begin{math}
        \cardin{\EE} = \sum_{\{B,C\} \in \WR} \pth{ \cardin{B} + \cardin{C}} = O( \eps^{-d} n \log \spread ).
    \end{math}

    Let $\G = (\P,\EE)$ be the resulting graph. Its maximum outdegree is $\Delta = O(\tfrac{1}{\eps^d} \log \spread)$, as a point has an outgoing edge for each pair it is in.  The exact query process is described above in \remref{search} -- it is the ``mature'' greedy routing with early stop. Specifically, the query process stops as soon as the improvement in the distance to the query fails to shrink by a factor of (say) $1-\eps/4$ from the previous iteration.

    Given a query point $q \in \Re^d$, let $p_1, \ldots, p_k$ be the sequence of vertices visited by the greedy walk in $\G$ for $q$ (here $p_1$ is picked arbitrarily). Let $t \in \P$ be the nearest neighbor to $q$ in $\P$, let $\ell = \dY{q}{t}$. and $\ell_i = \dY{p_i}{q}$, for $i =1,\ldots, k$.

    \begin{figure}[H]
        \centering
        \includegraphics{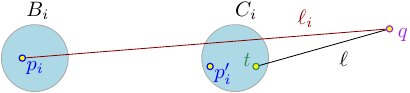}
        \caption{}
        \figlab{wspd}
    \end{figure}

    Let $\{B_i,C_i\}$ be the pair in the \WSPD covering the pair $p_i t$, and assume for concreteness that $p_i \in B_i$ and $t \in C_i$, and let $p_i' = \repX{C_i}$, see \figref{wspd}.  By the \WSPD property, we have that
    \begin{equation*}
        \dY{p_i' }{t}
        \leq
        \diamX{ C_i}
        \leq
        \frac{\eps}{8} \dsY{B_i}{C_i}
        \leq
        \frac{\eps}{8} \dY{p_i}{t}
        \leq
        \frac{\eps}{8}\pth{ \dY{p_i}{q} + \dY{q}{t}}
        =
        \frac{\eps}{8}\pth{ \ell_{i} + \ell}.
    \end{equation*}
    That implies, as $\dirEdgeY{p_i}{p_i'} \in \EGX{\G}$, that the algorithm considered $p_i'$ as its next stop after $p_i$. Namely, we have
    \begin{equation*}
        \ell_{i+1}
        \leq%
        \dY{p_i'}{q}
        \leq
        \dY{p_i'}{t}
        +
        \dY{t}{q}
        \leq
        \frac{\eps}{8}\pth{ \ell_{i} + \ell}
        + \ell
        \leq
        \frac{\eps}{4} \ell_{i}
        + \ell.
    \end{equation*}
    If the algorithm is not there yet, that is $\ell_i > (1+\eps)\ell$, then
    \begin{equation*}
        \ell_{i+1}
        \leq%
        \frac{\eps}{4} \ell_{i} + \ell
        <
        \pth{\frac{\eps}{4} + \frac{1}{1+\eps}} \ell_i
        \leq
        \pth{\frac{\eps}{4} + 1-\frac{\eps}{2}} \ell_i
        =%
        \pth{1-\frac{\eps}{4}} \ell_i,
    \end{equation*}
    and the query process would not stop in this iteration.

    The above already gives us a bound on the number of iterations in the query process. Indeed, if $\ell_1 > 4\diamX{\P}/\eps$, the algorithm stops immediately in the next round, as any point of $\P$ is the desired \ANN. Similarly, if $\ell_i < \cpX{\P}/2$, the algorithm found the nearest-neighbor point, and no further improvement in the current distance is possible, and the query process stops. As the distance shrinks by at least a factor of $1 - \eps/4$ at each iteration, it follows that the algorithm performs $O\Bigl(1+ \log_{1/(1-\eps/4)} \bigl(\diamX{\P}/(\eps \cpX{\P}) \bigr) \Bigr) = O( \Eps \log \spread )$ iterations (assuming that $\tfrac{1}{\eps} < \spread$).

    The above analysis can be further improved by observing that the distance shrinks more quickly during initial iterations.  Indeed, if $\ell_i > 8 \ell/\eps$, then $\ell_{i+1} \leq \eps \ell_i$.  If $\ell_i \in [ 3 \ell, 8\ell/\eps]$, then $\ell_{i+1} \leq 3\ell$. Finally, if $\ell_i \leq 3 \ell$, then
    \begin{math}
        \ell_{i+1}
        \leq%
        \frac{\eps}{4} \ell_{i}
        + \ell
        \leq (1+\eps)\ell.
    \end{math}
    But then, the algorithm terminates in the next few iteration, as $(1-\eps/4)^{10}\ell_{i+1} \leq (1-\eps)(1+\eps)\ell < \ell$, Namely, the algorithm performs $O( \frac{\log \spread}{\log(1/\eps) } + 1 )$ iterations, and each iteration takes $O(\Delta) = O( \eps^{-d} \log \spread)$ time.
\end{proof}

\begin{remark}
    It is tempting to further sparsify the above graph by connecting only representatives, as done in the spanner constructions using WSPD. There lies the rub -- while the new graph still has short paths to the nearest-neighbor, these paths are no longer direct or locally traceable. Computing these paths requires a ``higher-level'' approach. In the settings here,  the target vertex, $t$, is unknown -- all we have is the somewhat opaque information provided by the distance to the query point to guide the search.
\end{remark}

\subsection{Improving performance}

It is not hard to improve the above scheme, as described next. A natural approach is to build two graphs $\G_{1/2}$ and $\G_\eps$ -- the first uses $\eps = 1/2$, and the second uses the given value of $\eps$. We do the $1/2$-\NN greedy walk in $\G_{1/2}$, and then use the end vertex of this walk as a starting point for the $\eps$-\NN greedy walk in $\G_\eps$. This two-round approach yields the following result.

\begin{theorem}
    \thmlab{main_wspd}%
    Given a set $\P$ of $n$ points in $\Re^d$ with spread $\spread$, one can construct two graphs $\G_{1/2}, \G_\eps$ on $\P$, such that $(1+\eps)$-\ANN queries on $\P$ can be answered by first performing a $1/2$-\ANN greedy walk in $\G_{1/2}$, and then using the returned vertex as the starting vertex for a $(1+\eps)$-\ANN greedy walk in $\G_\eps$. The resulting point is a $(1+\eps)$-\ANN to the query point in $\P$, and the two walks take $O(\eps^{-d} \log \spread + \log^2 \spread)$ time overall. The two computed graphs have $O( \tfrac{1}{\eps^d} n \log \spread )$ edges overall.
\end{theorem}

\begin{proof}
    The key observation is that the first walk takes $O( \log^2 \spread )$ time, as the approximation factor is a constant. While the second walk involves at most two iterations, and thus takes $O(\eps^{-d} \log \spread)$ time.
\end{proof}

\subsubsection{An improved query time}

One can improve the query time even further by slicing the graphs. We point this out as an indication that the above scheme is probably not optimal, although the suggested scheme is a bit involved.

\begin{lemma}
    \lemlab{multi:graphs}%
    The query time of \thmref{main_wspd} can be improved to $O( \log \spread + \tfrac{1}{\eps^d} \log \tfrac{1}{\eps})$.
\end{lemma}
\begin{proof}
    The basic idea is to use the intuition from the previous analysis -- the length of edges used by the walk is exponentially decreasing till one gets close to the query. We can thus use this by limiting the algorithm to use only edges that are roughly in the current resolution. If these edges provide no improvement, the algorithm moves down to a lower resolution.

    Assume the closest-pair distance in $\P$ is $1$. We slice the graph $\G_{1/2}$ into graphs $\H_1, \ldots, \H_m$, where $m = \ceil{\log_2 \spread}$, and $\H_i$ contains all the edges of $\G_{1/2}$ of length in the range $[\spread/2^{i+3},\spread/2^{i-3}]$. It is easy to verify that the degree of each vertex in the graph $H_i$ is $O(1/(1/2)^d) = O(1)$. Note that an edge of $\G_{1/2}$ appears in $7$ of the slice graphs. Now, the idea is to start the $\tfrac{1}{2}$-\NN greedy walk in $\H_i$ for $i=1$.  As soon as it gets stuck, the algorithm moves the walk to $\H_{i+1}$, and continues until it arrives at $\H_m$. Assume that this process ended at $p \in \P$, with $L = \dY{q}{p}$, where $q$ is the query point.

    We now repeat the same slicing idea for $\G_\eps$, and start the walk from $p$ in the first sliced graph that contains edges of length $L$. It is easy to verify that the walk now would use only $O( \log \tfrac{1}{\eps} )$ of these graphs till the length of the edges becomes so small that the search stops, and the desired $(1+\eps)$-\ANN is computed.

   Putting everything together, the resulting running time is $O( \log \spread + \tfrac{1}{\eps^d} \log \tfrac{1}{\eps})$.
\end{proof}

\section{A \NN graph via greedy permutation}
\seclab{n_n_greedy}

\subsection{Background: Greedy permutation.}
\seclab{greedy_permutation}

Given a finite metric space $\FMS = (\P, \dC)$, a \emphi{$\kappa$-greedy permutation}, for some $\kappa \geq 1$, is an ordering $p_1, \ldots, p_n$ of the points of $\P$, with associated radii $r_1 \geq r_2 \geq \cdots \geq r_{n+1}$, such that:
\begin{compactenumA}
    \medskip%
    \item The point $p_1$ is an arbitrary point of $\P$, and $r_1 = \max_{ p\in \P}\dY{p}{p_1}$.

    \medskip%
    \item For all $i \in \IRX{n}= \{ 1,\ldots, n\}$, all the points of $\P$ are covered by the union of balls of radius $\kappa r_{i}$ centered at the points of $\P_{i} = \{ p_1, \ldots, p_i\}$ -- formally, $\P \subseteq \cup_{j=1}^i \ballY{p_j}{\kappa r_i}$.

    \medskip%
    \item For all $i >1$, the distance of $p_i$ from $\P_{i-1}$ is $r_{i-1}$, and furthermore, $\cpX{\P_{i}} = r_{i-1}$ (i.e., the closest-pair distance in $\P_i$ is $r_{i-1}$).
\end{compactenumA}

\begin{observation}
    \obslab{exact}%
    The algorithm that repeatedly picks the furthest point $p_i \in \P \setminus \P_{i-1}$ from $\P_{i-1}$, and adds it to the greedy permutation, computes it exactly (i.e., $\kappa=1$), in quadratic time. The exact greedy permutation is a packing for all prefixes: That is, for all $i$, the set $\P_i$ is an $r_i$-packing\footnote{Nit-packing a bit, it is an $r_{i-1}$-packing, see \defref{packing}, assuming that all pairwise distances in $\P$ are unique.} of $\P$, see \defref{packing}.
\end{observation}

For a set $\P$ of $n$ points in $\Re^d$ (or in a metric space of bounded doubling dimension), Har-Peled and Mendel \cite{hm-fcnld-06} showed how to compute the $\kappa$-greedy permutation in $O(n \log n)$ time, where $\kappa = 1+ 1/n^{O(1)}$. We assume that the exact greedy permutation is available for simplicity of exposition.

An additional useful property of the algorithm of Har-Peled and Mendel is that, for all $i$, one can compute for each point $p_i$, all the points of $\P_{i-1}$ in distance at most (say) $4r_{i-1}/\eps$ from it. Formally, let
\begin{equation}
    F_i = \P_{i-1} \cap \ballY{p_{i}}{8 r_{i-1}/\eps}
    \eqlab{friends}%
\end{equation}
be the \emphi{friend list} of $p_i$ (the friend list definition in \cite{hm-fcnld-06} is roughly the same when $\eps > 1/4$, otherwise one needs to perform a local traversal on the net-tree, to compute $F_i$, that takes $O( |F_i| )$ time). Intuitively, the friend list of $p_i$ is the set of all the points, in the packing $\P_{i-1}$, that are relatively close to $p_i$.

Since $\P_{i-1}$ is a $r_{i-1}$-packing, if we place a ball of radius $r_{i-1}/2$ around each point of $\P_{i-1}$, they would all be interior disjoint. As such, for all $p \in \Re^d$ and $R > 0$, we have that
\begin{equation*}
    \cardin{\ballY{p}{R} \cap \P_{i-1}}
    =
    O\bigl( (1 + R/r_{i-1})^d \bigr).
\end{equation*}
Thus, we have $\cardin{F_i} = O(1/\eps^d)$ for all $i$.  Observe that for all $p_j \in F_i$, we have $j < i$.

\subsection{The graph construction}

Given a set $\P$ of $n$ points in $\Re^d$, and a parameter $\eps \in (0,1/2)$, the algorithm first computes the greedy permutation of $\P$, and the friends list of each point, as described above. Next, the algorithm builds a directed graph $\G =(\P,\EE)$, with the edges being
\begin{equation*}
    \EE
    =
    \SetY{ \dirEdgeY{p_j}{p_i}}{ p_j \in F_i, \text{ for } i =1,\ldots, n}.
\end{equation*}
In the constructed graph, the list of outgoing edges $\EE_v$, from a vertex $v$, is sorted in increasing order by the index of the destination. This ordering can be realized by always adding the outgoing edges at the end of this list.

\paragraph{Answering \ANN queries.}
The search uses the ``impulsive'' greedy routing described in \secref{greedy_routing}. Given a query point $q \in \Re^d$, the algorithm starts with the current point being $c=p_1$. The algorithm now scans the outgoing edges $\dirEdgeY{c}{p_j}$ from the current vertex, sorted by increasing index $j$. The algorithm sets $c=p_j$, as soon as an edge $\dirEdgeY{c}{p_j}$ is encountered such that
\begin{equation*}
    \dY{q}{p_j} \leq (1-\eps/4) \dY{q}{c}.
\end{equation*}
It then restarts the scanning process of the outgoing edges of the new vertex $c$. This process continues until all the outgoing edges of the current vertex have been scanned without finding a profitable move, and the algorithm returns the current node.

\subsection{Analysis}

Clearly the graph $\G$ has $O( n/\eps^d)$ edges, as the $i$\th vertex has at most $|F_i| = O(1/\eps^d)$ incoming edges.

\begin{observation}
    Consider a distance $L>0$, and points $p_j,p_i \in \P$.  An edge $\dirEdgeY{p_j}{p_i} \in \EGX{\G}$, with $j < i$, is \emphw{$L$-admissible} if $\dY{p_j}{p_i} \in [L/2,L]$. This implies that the radius $r_{i-1} =\dmY{p_i}{\P_{i-1}}= \Omega(\eps L)$. Otherwise, $p_j$ is too far from $p_i$ to be connected to it, and the edge would not be present in $\G$. Formally, assume that $r_{i-1}< \eps L /16$, and observe that all the edges incoming into $p_i$ can have length at most  $8r_{i-1}/\eps < L/2$,
    by \Eqref{friends}, which is a contradiction to the edge being $L$-admissible.

    We claim that at most $O( 1/\eps^d)$ \ensuremath{L}-admissible edges emanating from a vertex $p \in \P$.  Indeed, let $p_j$ be the last vertex such that $\dirEdgeY{p}{p_j}$ is $L$-admissible. Then, by the above, $r_{j-1} = \Omega( \eps L)$. Namely, $P_{j}$ is an $\Omega(\eps L)$-packing, and it can contain at most $O(1/\eps^d)$ points in the ball centered at $p$ of radius $L$.

    Let $\bigl.n(p, L)$ denote the number of $L$-admissible edges for $p$. The total number of out-edges of $p$ is at most
    \begin{math}
        \sum_{i=0}^{\log\spread} n( p, \diamX{\P}/2^i ) = O\pth{\eps^{-d} \log \spread},
    \end{math}
    where $\spread$ is the \defrefY{spread}{spread} of $\P$.
\end{observation}

\begin{figure}
    \centering%
    \includegraphics{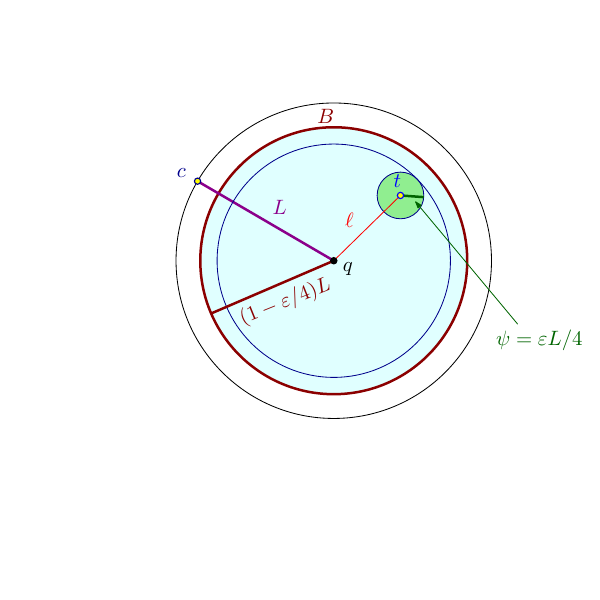}
    \caption{Illustration of proof.}
    \figlab{good_ball}%
\end{figure}

\begin{lemma}
    \lemlab{walk}%
    For any query point $q \in \Re^d$, the greedy routing for $q$ (starting from $p_1$), in the above constructed graph $\G$, returns a point $p \in \P$, such that $\dY{q}{p} \leq (1+\eps)\dmY{q}{\P}$.  The query time is $O( \eps^{-d-1} \log^2 \spread)$, where $\spread=\spreadX{\P}$.
\end{lemma}
\begin{proof}
    Assume the algorithm just moved to the point $p_j \in \P$ and let $L = \dY{q}{p_j}$. Let $t$ be the nearest-neighbor to $q$ in $\P$, with $\ell = \dY{q}{t}$. If $L \leq (1+\eps) \ell$, the algorithm gets the desired \ANN and returns. Otherwise, $L > (1+\eps) \ell$ and
    \begin{equation*}
        \Bigl(1- \frac{\eps}{4}\Bigr) L - \frac{\eps}{4} L
        =
        L (1-\tfrac{\eps}{2})
        \geq
        (1+\eps)(1-\tfrac{\eps}{2})  \ell
        \geq%
        (1+\tfrac{\eps}{4})  \ell
        >
        \ell,
    \end{equation*}
    since $\eps \leq 1/2$.  Namely, for $\psi = \eps L /4$, the ball $\ballC = \ballY{t}{\psi}$ is fully contained inside the ball $B=\ballY{q}{(1-\eps/4)L}$, see \figref{good_ball}.

    The algorithm has not scanned any point in $B \cap \P$. Indeed, if it had scanned such a point, it would have moved to this point. Let $p_\alpha$ be the point in $\P$, with minimum index $\alpha$, such that $p_\alpha$ is in the ``small'' ball $\ballC$.  Assume for the time being that $\alpha > j$ (i.e., the current point $c=p_j$). We have
    \begin{equation*}
        r_{\alpha-1}
        \geq
        r_\alpha =
        \dsY{p_\alpha}{\P_{\alpha-1}}
        \geq
        \dsY{t}{\P_{\alpha-1}}
        \geq
        \psi,
    \end{equation*}
    as $p_\alpha$ is the furthest point in $\P$ from $\P_{\alpha-1}$ (by the greedy permutation construction), and the ball $\ballC$ does not contain any point of $\P_{\alpha-1}$, see \obsref{exact}. We conclude that
    \begin{equation*}
        \dY{p_j}{p_\alpha}
        \leq
        \dY{p_j}{q}
        + \dY{q}{t}
        + \dY{t}{p_\alpha}
        \leq
        2 L
        =%
        \frac{8}{\eps} \psi
        \leq
        \frac{8}{\eps} r_{\alpha-1}.
    \end{equation*}
    But then the algorithm added the edge $\dirEdgeY{p_j}{p_\alpha}$ to $\G$ during its construction, see \Eqref{friends}.  Furthermore,
    \begin{equation*}
        \dY{q}{p_{\alpha}}
        \leq
        \dY{q}{t} + \dY{t}{p_{\alpha}}
        \leq
        \ell + \frac{\eps}{4}L
        \leq
        \bigl(1-\frac{\eps}{4}\bigr)L.
    \end{equation*}
    Thus, either the algorithm moved to $p_\alpha$, or some other close point to $q$. Namely, the distance of the point the algorithm moved to after $p_j$ had decreased the distance to $q$ by a factor of (at least) $1-\eps/4$.

    If $\alpha < j$, consider the last point $p_\beta$ that the algorithm moved to (before moving to $p_j$) with $\beta < \alpha$. But then, the same argument as above shows that $\dirEdgeY{p_\beta}{p_\alpha} \in \EGX{\G}$, see \remref{elaborate} below. Namely, the algorithm must have moved to $p_{\alpha}$, and thus never moved to $p_j$, which is a contradiction.

    The number of steps performed by the algorithm is $O( \log_{1/(1-\eps/4)} \spreadX{\P} ) = O( \eps^{-1} \log \spread)$. Each scan naively takes $O( \eps^{-d} \log \spread )$ time, which bounds the maximum outdegree in the graph $\G$, thus implying the stated bound.
\end{proof}

\begin{remark}
    \remlab{elaborate}%
    We elaborate here on the ``same'' argument above. The algorithm visited a vertex $p_\beta$, then took an edge to a later vertex $p_\gamma$, such that $\beta < \alpha < \gamma$ (i.e., the algorithm skipped\footnote{Or ``overflew'' $p_\alpha$, recalling a memorable excuse why a commercial flight one of the authors took did not land in its stated midway destination.} $p_\alpha$), on its way to the current vertex $p_j$. As a reminder, $p_\alpha$ is the first point (in the permutation) in $\ballC$. We have $L^+ = \dY{q}{p_\beta} > \dY{q}{p_j} = L$. Let $\psi^+ = \eps L^+ /4$, and observe that $\ballC \subseteq \ballC^+ = \ballY{t}{\psi^+} \subseteq B^+ =\ballY{q}{(1-\eps/4)L^+}$. But then, $r_{\alpha-1} \geq \psi^+$, and (arguing as above) the edge $\dirEdgeY{p_\beta}{ p_\alpha}$ is in the graph, and the search algorithm is forced to take it when scanning the outgoing edges of $p_\beta$ -- a contradiction.
\end{remark}

\paragraph{Improving the query process.}

We rebuild the above graph so that it answers $(1+\eps/4)$-\ANN queries. The above algorithm is \emphw{forward scanning} -- if an edge $\dirEdgeY{p_j}{p_i}$ is inspected by the algorithm, all future edges $\dirEdgeY{p_u}{p_v}$ inspected by the algorithm would have $v \geq i$ (we also have that $u=j$ or $u \geq i$).

The idea is to modify the algorithm so that it terminates early.
\begin{claim}
    \clmlab{early_stop}%
    If $e = \dirEdgeY{p_j}{p_i}$ is inspected by the algorithm, $\dY{q}{p_j} < \dY{q}{p_i}$, and $r_i < (\eps/8)\dY{q}{p_j}$, then $p_j$ is $(1+\eps)$-\ANN to $q$ in $\P$, and the algorithm can stop.
\end{claim}
\begin{proof}
    Since $\P_i$ is an $r_i$-packing of $\P$, there must be a point $p' \in \P_i$ that is in distance $r_i$ from $t$, where $t = \nnY{q}{\P}$. We can interpret the algorithm as working on $\P_i$ (instead of $\P$). Indeed, the induced subgraph on $\P_i$, $\G_i = \G_{\P_i}$, is the same as the graph the algorithm would build if the input point set is $\P_i$. The query process on $\G_i$ is identical to the one on $\G$, as long as we inspect edges in $\G_i$.  Thus, if we run the algorithm on $\G_i$, $e$ is the last edge inspected. But $p_i$ is not an improvement, so $p_j$ is the point the algorithm returns when run on $\G_i$.  \lemref{walk} then implies that $p_j$ is $(1+\eps/4)$-\ANN (as we calibrated $\eps$ to be $\eps/4$).

    Thus, we have that  $\nu = \dY{q}{p_j} \leq (1+\eps/4)\dmY{q}{\P_i}$, and
    \begin{equation*}
        r_i
        <%
        \frac{\eps}{8}\dY{q}{p_j}
        \leq
        \frac{\eps}{8}\bigl(1+\frac{\eps}{4}\bigr)\dmY{q}{\P_i}
        \leq%
        \frac{\eps}{4}\pth{\dmY{q}{\P} + r_i}
        \implies
        r_i
        \leq
        \frac{\eps}{4(1-\eps/4)}\dmY{q}{\P}
        \leq
        \frac{\eps}{3}\dmY{q}{\P}.
    \end{equation*}
    We conclude that
    \begin{equation*}
        \dY{q}{p_j}
        \leq%
        (1+\eps/4)\dmY{q}{\P_i}
        \leq
        (1+\eps/4)\pth{\dmY{q}{\P} + r_i}
        \leq
        \pth{1+\frac{\eps}{4}}
        \pth{1+\frac{\eps}{3}}
        \dmY{q}{\P}
        \leq
        \pth{1+\eps}
        \dmY{q}{\P}.
    \end{equation*}
\end{proof}

Let $\diamC = \diamX{\P}$, and let
\begin{equation*}
    R_i = \frac{\diamC}{2^i   }
\end{equation*}
for $i=0,1,\ldots, h$, where $h = \ceil{\log_2 \spreadX{\P}}$.  Consider the greedy permutation $p_1, \ldots, p_n$, and the associated radii $r_1 \geq r_2 \geq \cdots \geq r_n$. The \emphi{$i$\th epoch} of $\P$, is a block $B_i = \permut{ p_\alpha, \ldots, p_\beta}$, such that $\alpha < \beta$, $|\beta -\alpha|$ is maximal, and $r_\alpha, r_{\alpha+1}, \ldots, r_\beta \in [R_i,R_{i-1})$.

\begin{lemma}
    When using early stop, the query time of the \ANN algorithm is at most $O( \eps^{-d} \log \spread)$.
\end{lemma}
\begin{proof}
    The algorithm's running time is proportional to the number of edges $\dirEdgeY{p_j}{p_i}$ it scans. There could be at most $O( \eps^{-1} \log \spread)$ edges that cause the algorithm to change the current vertex, as each such change decreases the \NN distance by a factor of $1-O(\eps)$.

    So we only have to pay for edges scanned in vain, without triggering a change to the current vertex. And let $E_i$ be all these edges whose destination is in the $i$\th epoch $B_i$, and let $V_i \subseteq B_i$ be the set of destinations of the edges of $E_i$.

    Let $x \rightarrow y$ be the first edge of $E_i$ scanned, and let $L_i = \dY{q}{x}$.  \clmref{early_stop} implies that $L_i = O( R_i/\eps)$ (as otherwise the algorithm would have terminated). But then, all the points of $V_i$ are contained inside the ball $\ballY{q}{2L_i}$. Since these points are all at a distance of at least $R_i/2$ from each other, it follows that
    \begin{equation*}
        \cardin{V_i} = \cardin{B_i \cap \ballY{q}{2L_i} } = O(1/\eps^d).
    \end{equation*}
    Since there are $O( \log \spread )$ epochs, it follows that the total number of edges scanned in vain is $O( \eps^{-d} \log \spread)$, which also bounds the running time.
\end{proof}

\begin{theorem}
    \thmlab{a_n_n_main}%
    Given a set $\P$ of $n$ points in $\Re^d$, and a parameter $\eps \in (0,1)$, one can construct a directed graph $\G = (\P, \EE)$ with $O( n /\eps^d)$ edges, such that given a query point $q$, one can compute a $(1+\eps)$-\ANN to $q$ by performing a greedy $\eps$-\NN walk in $\G$. This walk takes $O( \eps^{-d} \log \spread)$ time, where $\spread$ is the spread of $\P$.
\end{theorem}

\begin{remark}
    The result of \thmref{a_n_n_main} holds if $\P \subseteq \Ground$ is a set of $n$ points in a metric space $\FMS = (\Ground, \DistChar )$ of bounded doubling dimension $\delta$. The term $d$ is then replaced by $O(\delta)$. Thus, the space of the construction is $n / \eps^{O(\delta)}$, and the query time is  $ \eps^{-O(\delta)} \log \spread$.
\end{remark}

\BibLatexMode{\printbibliography}

\appendix

\section{Apollonius circle}

\begin{lemma}
    \lemlab{Apollonius}%
    Let $u_1,u_2$ be two points in $\Re^d$, and consider the set $U$ of all points $p \in \Re^d$, such that $w_1 \dY{u_1}{p} \geq w_2 \dY{u_2}{p}$, where $w_1, w_2$ are two specified weights. For $\xi = \dY{u_1}{u_2}$, the set $U$ is the Apollonius ball centered at
    \begin{equation*}
        u_2 +
        \frac{1}{\kappa^2 - 1}  (u_2-u_1),
    \end{equation*}
    and of radius
    \begin{math}
        \frac{\kappa}{\kappa^2 - 1} \xi.
    \end{math}
\end{lemma}
\begin{proof}
    By rotating and translating space, we can assume that  $u_1=(0,0)$ and $u_2=(\xi,0)$ be two points, with weights $w_1$ and $w_2$, respectively. The Apollonius circle they define is
    \begin{equation*}
        w_1 \dY{u_1}{(x,y)} = w_2\dY{u_2}{(x,y)}.
    \end{equation*}
    Setting $\kappa = w_2/w_1$, and squaring, we have
    \begin{align*}
      &x^2 +  y^2
        =
        \kappa^2 \pth{ (x - \xi )^2 + y^2}
      \\%
      \iff\quad
      &
        0=
        (\kappa^2-1)\pth{ x^2 + y^2}
        +
        \kappa^2 \pth{ -2 \xi x + \xi^2
        }
      \\
      \iff\qquad
      &
        0=
        x^2 -2 \frac{\kappa^2}{\kappa^2 - 1} \xi x
        + y^2
        +
        \frac{ \kappa^2   }{\kappa^2-1}    \xi^2
      \\
      \iff\qquad
      &
        \pth{x - \frac{\kappa^2}{\kappa^2 - 1} \xi}^2
        + y^2
        =
        \pth{\frac{\kappa^2}{\kappa^2 - 1} \xi}^2
        -\frac{ \kappa^2   }{\kappa^2-1}    \xi^2
        =
        \pth{\frac{\kappa}{\kappa^2 - 1} \xi}^2,
    \end{align*}
    since
    \begin{math}
        {\frac{\kappa^2}{\kappa^2 - 1} \xi^2}\pth{
           \frac{\kappa^2}{\kappa^2 - 1}
           - 1 }
        =
        {\frac{\kappa^2}{\kappa^2 - 1} \xi^2}\pth{
           \frac{1}{\kappa^2 - 1}
        }
        =
        \frac{\kappa^2}{(\kappa^2 - 1)^2} \xi^2.
    \end{math}
    Namely, the disk has a center at
    \begin{equation*}
        \pth{\frac{\kappa^2}{\kappa^2 - 1} \xi, 0}
        =%
        u_1 + \frac{\kappa^2}{\kappa^2 - 1}  (u_2-u_1)
        =%
        u_1 + \pth{ 1 + \frac{1}{\kappa^2 - 1} } (u_2-u_1)
        =
        u_2 +
        \frac{1}{\kappa^2 - 1}  (u_2-u_1).
    \end{equation*}
    and its radius is
    \begin{math}
        r =
        \frac{\kappa}{\kappa^2 - 1} \xi.
    \end{math}
\end{proof}

\end{document}